\documentclass[pra,floatfix,amsmath,superscriptaddress,tightenlines,twocolumn,nofootinbib]{revtex4-1}
\usepackage{amsmath, mathtools,amssymb,bbm,amsthm}
\usepackage{tikz}
\usepackage{graphicx}
\usepackage{soul}
\usepackage{float}
\usepackage{mathtools}
\usepackage{braket}
\usepackage{booktabs}
\usepackage{algorithmicx}
\usepackage{jeffe}
\usepackage{algorithm}
\usepackage[noend]{algpseudocode}
\usepackage[pdftex, pdftitle={Article}, pdfauthor={Author}, linktocpage]{hyperref}
\usepackage[]{qcircuit}

\usepackage{etoolbox}

\newcommand{\Tr}{\text{Tr}}

\newcommand{\cT}{\mathcal{T}}
\newcommand{\cU}{\mathcal{U}}
\newcommand{\cH}{\mathcal{H}}

\newcommand{\cI}{\mathcal{I}}

\newcommand{\cE}{\mathcal{E}}

\newtheorem{definition}{Definition}
\newtheorem{theorem}{Theorem}[section]
\newtheorem{corollary}[theorem]{Corollary}
\newtheorem{lemma}{Lemma}
\newtheorem{prop}{Proposition}

\begin{document}
\title{Recycling qubits in near-term quantum computers}
	
\author{Galit Anikeeva}
\affiliation{Department of Physics, Stanford University, Stanford, CA 94305, USA}
\author{Isaac H. Kim}
\affiliation{Centre for Engineered Quantum Systems, School of Physics, University of Sydney, Sydney, NSW 2006, Australia}
\affiliation{Stanford Institute for Theoretical Physics, Stanford University, Stanford, CA 94305, USA}
\author{Patrick Hayden}
\affiliation{Stanford Institute for Theoretical Physics, Stanford University, Stanford, CA 94305, USA}
\date{\today}

\begin{abstract}
Quantum computers are capable of efficiently contracting unitary tensor networks, a task that is likely to remain difficult for classical computers. For instance, networks based on matrix product states or the multi-scale entanglement renormalization ansatz (MERA) can be contracted on a small quantum computer to aid the simulation of a large quantum system. However, without the ability to selectively reset qubits, the associated spatial cost can be exorbitant. In this paper, we propose a protocol that can unitarily reset qubits when the circuit has a common convolutional form, thus dramatically reducing the spatial cost for implementing the contraction algorithm on general near-term quantum computers. This protocol generates fresh qubits from used ones by partially applying the time-reversed quantum circuit over qubits that are no longer in use. In the absence of noise, we prove that the state of a subset of these qubits becomes $|0\ldots 0\rangle$, up to an error exponentially small in the number of gates applied. We also provide a numerical evidence that the protocol works in the presence of noise, and formulate a condition under which the noise-resilience follows rigorously.
\end{abstract}

\maketitle

\section{Introduction}

A quarter century after the discovery that quantum computers could efficiently solve classically intractable  problems~\cite{Shor1997,Lloyd1996}, there is now a worldwide academic and industrial effort to build machines realizing that vision. An important milestone was passed last year with Google's experimental demonstration of quantum supremacy using a 54 qubit processor~\cite{Arute2019}. While the quantum supremacy experiment demonstrated that there are \emph{some} classically intractable tasks that a quantum computer can efficiently perform, it is still unclear if a speedup can be realized in the near term for any practical problems of interest. One promising approach is to pursue variational quantum algorithms~\cite{Peruzzo2014,Farhi2014}. These algorithms use a quantum computer as a device that can prepare quantum states, parametrized by a circuit of unitary operators. By variationally minimizing the energy of a physical system of interest, one can prepare a low-energy state of a quantum many-body system. If the states prepared in this minimization procedure are difficult to simulate on a classical computer, one may hope to obtain a quantum advantage in finding and characterizing the low-energy states. 

However, to achieve a practical advantage, one needs to overcome several challenges. For generic circuits, the gradient decays exponentially with circuit depth~\cite{McClean2018}, making gradient-based optimization challenging. Moreover, realistic devices are noisy and, as such, any intermediate or final outcomes obtained from these experiments will be noisy as well. These limitations suggest that, in order to use variational algorithms for practical purposes, one must use a structured circuit such as the ones introduced in Ref.~\cite{Wecker2015,Kim2017,Kim2017b,Cong2019,Bausch2006}.

The family of \textit{convolutional quantum circuits}~\cite{poulin2009quantum}, in which a sliding active window moves over the qubits of the circuit, is particularly promising. The family includes circuits  referred to in previous literature as \emph{holographic quantum circuits}~\cite{Kim2017,Kim2017a,Kim2017b,Borregaard2019,Barratt2020,Olund2020,Foss-Feig2020}), which can be used to realize versions of the matrix product state~\cite{cirac2020matrix}, projected entangled pair states~\cite{peps}, and MERA~\cite{Vidal2007}. Holographic circuits have a number of desirable properties. First, there are theoretical arguments for the faithfulness of the circuits in describing physical states of interest~\cite{Kim2017,Haegeman2018}. Second, generically, the expectation values of the local observables obtained from these circuits are resilient to noise~\cite{Kim2017a,Kim2017b}. Third, these circuits often lead to a reduced spatial cost. For instance, a quantum many-body system in $D$ spatial dimensions can be simulated using a set of qubits arranged on a $(D-1)$-dimensional lattice. Therefore, the existing quantum computing architectures, which are either one- or two-dimensional, can simulate two- and three-dimensional quantum many-body systems. If the target system one intends to simulate has a volume of $\ell^D$, the number of qubits one would need to carry out such a simulation would scale as $\mathcal{O}(\ell^{D-1})$. Indeed, a recent experiment demonstrates the success of this approach for $D=1$, giving hope that these methods can become useful with a continued improvement in quantum computing technology~\cite{Barratt2020,Foss-Feig2020}. 

Despite their promise, convolutional quantum circuits have not been widely used on existing quantum computing platforms. This is mainly because they require a source of fresh $\ket{0}$ qubits to realize their advantages, but resetting qubits in the middle of a computation can be demanding. For example, the Sycamore processor which Google used in their quantum supremacy experiment does not have the ability to quickly reset a specific qubit~\cite{Arute2019}. Without that capability, the size of the system one can simulate is limited by the number of qubits present in the quantum computer. For instance, even if one uses convolutional quantum circuits for simulating a $\ell \times \ell$ system, without the ability to reset, one would still need a quantum computer consisting of $\mathcal{O}(\ell^2)$ qubits. In contrast, with the ability to reset, one can reduce this qubit requirement to $\mathcal{O}(\ell)$. For instance, if $\ell=10$, the requirement on the number of qubits can shrink from hundreds to tens.

Since having the reset capability can dramatically reduce the number of physical qubits required, it is clearly of interest to implement it in some form. This article provides a software-based reset protocol suitable for systems in which the hardware capability is either absent or too slow. Our procedure, which we refer to as the \emph{rewinding protocol}, unitarily converts a set of qubits used in the computation to a fixed state, after which they can be reused in the remaining part of the computation. Because our protocol is unitary, the protocol can be implemented on any quantum computer, even if it is not equipped with a physical reset operation. 

Importantly, we can provide a rigorous theoretical guarantee on the quality of the reset qubits. The infidelity of the reset qubits decays exponentially with the length of the protocol. Therefore, resetting a qubit up to a fixed target error tolerance takes time logarithmic in the inverse error. 

While other strategies exist for reusing qubits, they have limitations that make them unsuitable for our purpose. As stated before, in certain experimental platforms, qubit reset is unavailable. Alternatively, one could resort to algorithmic cooling~\cite{Fernandez2004,Schulman2005}. However, algorithmic cooling can fail if the state of the qubits has off-diagonal components in the computational basis or if the qubits are correlated. Neither of these possibilities can be ruled out in our setup. 

In contrast, our protocol enjoys the following properties. First, our protocol is unitary. Therefore, it can be implemented on any quantum computing platform. Second, it successfully produces fresh qubits even in a regime where algorithmic cooling fails. Third, our protocol is easy to implement. The protocol involves partially applying the time-reversed quantum circuit over the qubits that are no longer in use, a capability that any reasonable quantum computer is likely to possess. There is no need for measurement, reset or new unitary gates that are incompatible with the given hardware. Lastly, we numerically demonstrate that our protocol works in the presence of noise. While we do not have a theoretical understanding of this phenomenon, our numerical evidence seems compelling.

The rest of this paper is structured as follows. In Section~\ref{sec:protocol}, we introduce our protocol. In Section~\ref{sec:nrqc}, we provide rigorous performance guarantee of our protocol, demonstrating its robustness to noise in Section~\ref{sec:discussion}.  In Section~\ref{sec:boost}, we describe and analyze a method for using our protocol to effectively boost the number of qubits available for implementing a convolutional circuit. We conclude with a discussion in Section~\ref{sec:discussion}.

\section{Rewinding protocol}
\label{sec:protocol}
In this section, we introduce the rewinding protocol. Without loss of generality, a general quantum computation can be thought of as a sequence of gates applied to the $|0^n\rangle$ state, where $n$ is the number of qubits. We shall refer to this sequence as 
\begin{equation}
    \mathcal{C} = (g_1, g_2, \ldots, g_{N-1}, g_N),
\end{equation}
where $g_i$ belongs to some universal gate set for every $i$. This circuit implements a unitary
\begin{equation}
    U(\mathcal{C}) := g_N g_{N-1}\cdots g_2g_1.
\end{equation}

In the quantum circuit, some of the qubits may no longer participate in the computation after a certain point. Specifically, we may be able to break up $\mathcal{C}$ into 
\begin{equation}
\begin{aligned}
    \mathcal{C}_1 &= (g_1, \ldots, g_{k-1}), \\
    \mathcal{C}_2 &= (g_k, \ldots, g_N),
\end{aligned}
\end{equation}
such that $\mathcal{C}_2$ acts trivially on a subset of qubits. For such qubits, we will say that they are \emph{idle} during $\mathcal{C}_2$.

The main goal of our protocol is to convert some of these idle qubits to $|0^m\rangle$ for some $m>0$. We propose to achieve this by implementing the following ``rewinding" circuit. Let $\mathcal{Q}_I$ be the set of idle qubits. Consider a subsequence of $\mathcal{C}_1$ consisting of gates that act exclusively on $\mathcal{Q}_I$. Let us refer to this circuit as 
\begin{equation}
    \mathcal{C}_{1,I} := (\tilde{g}_1, \ldots, \tilde{g}_{N'}).
\end{equation}
The rewinding circuit is defined as follows:
\begin{equation}
    \mathcal{R}(\mathcal{C}_{1,I}) := (\tilde{g}_{N'}^{\dagger}, \ldots, \tilde{g}_1^{\dagger}).
\end{equation}
After applying the rewinding circuit (see Fig.~\ref{fig:rewinding_general} for an example), we obtain the following state:
\begin{equation}
    \mathcal{R}(\mathcal{C}_{1,I}) U(\mathcal{C})|0\ldots 0\rangle.
\end{equation}
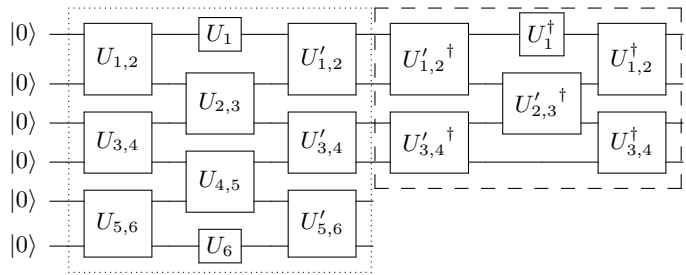
\begin{figure*}[t]
    \centering
    $
     \Qcircuit @C=.7em @R=.7em {
     \lstick{\ket{0}} & \qw & \multigate{1}{U_{1,2}} &\qw &\gate{U_1} & \qw &\multigate{1}{U_{1,2}'} &\qw &\multigate{1}{{U_{1,2}'}^{\dagger}}  & \qw &\gate{U_1^{\dagger}} &\multigate{1}{U_{1,2}^{\dagger}} &\qw\\
     \lstick{\ket{0}} & \qw & \ghost{U_{1,2}} & \qw &\multigate{1}{U_{2,3}} &\qw &\ghost{U_{1,2}'} &\qw &\ghost{{U_{1,2}'}^{\dagger}}  &\qw & \multigate{1}{{U_{2,3}'}^{\dagger}} & \ghost{U_{1,2}^{\dagger}} &\qw\\
     \lstick{\ket{0}} & \qw & \multigate{1}{U_{3,4}} & \qw & \ghost{U_{2,3}} &\qw &\multigate{1}{U_{3,4}'} &\qw &\multigate{1}{{U_{3,4}'}^{\dagger}} &\qw &\ghost{{U_{2,3}'}^{\dagger}} & \multigate{1}{U_{3,4}^{\dagger}} &\qw\\
     \lstick{\ket{0}} & \qw & \ghost{U_{3,4}} & \qw & \multigate{1}{U_{4,5}} &\qw &\ghost{U_{3,4}'} &\qw &\ghost{{U_{3,4}'}^{\dagger}} &\qw &\qw  & \ghost{U_{3,4}^{\dagger}} &\qw \\
     \lstick{\ket{0}} & \qw & \multigate{1}{U_{5,6}} & \qw &\ghost{U_{4,5}} &\qw &\multigate{1}{U_{5,6}'} &\qw \\
     \lstick{\ket{0}} & \qw & \ghost{U_{5,6}} &\qw & \gate{U_6} &\qw &\ghost{U_{5,6}'} &\qw
     \gategroup{1}{9}{4}{12}{1.25em}{--}
     \gategroup{1}{3}{6}{7}{1.25em}{.}
     }
     $
    \caption{An example of the rewinding circuit. Here the first $4$ qubits are assumed to be idle. The subcircuit in the dashed box is the rewinding circuit of the circuit in dotted box.}
    \label{fig:rewinding_general}
\end{figure*}

An important question is whether the state of the idle qubits becomes close to the $|0\ldots 0\rangle_I$ state in some sense. We formulate a sufficient condition under which this is possible. Of course, generally speaking, we should not expect this state to be close to $|0\ldots 0\rangle$. However, for a special class of circuits called \emph{convolutional circuits}~\cite{Kim2017,Kim2017a,Kim2017b}, one can guarantee this generically.

\section{Convolutional circuits}
\label{sec:nrqc}

In this section, we study how well the rewinding protocol works for convolutional circuits~\cite{poulin2009quantum}. These circuits have a fixed-size active window that moves over the qubits of the circuit. More formally, a convolutional circuit is one which can be written in the following form:
\begin{equation}
    U_{[n-k-1, n]}\cdots U_{[2,k+1]}U_{[1,k]},
\end{equation}
over $n$ qubits, where $U_{[i,j]}$ is a unitary operator acting on a set of qubits ranging from the $i$'th to the $j$'th index. Here $k$ is a constant that is often chosen to be $\mathcal{O}(1)$.

Convolutional circuits play an important role in near-term quantum algorithms because they naturally appear in quantum circuits that can approximately prepare physical ground states of interest~\cite{Kim2017,Kim2017a,Kim2017b,Haegeman2018}; see Appendix~\ref{sec:convolutional_circuits_examples} for a review. Specifically, for the purpose of evaluating the expectation values of local observables, which is a key subroutine used in many near-term quantum algorithms, one can compress those circuits into a convolutional circuit. Without loss of generality, suppose our goal is to evaluate
\begin{equation}
    \langle \psi| O |\psi\rangle
\end{equation}
where $|\psi\rangle$ is the state prepared by the circuits in Ref.~\cite{Kim2017,Kim2017a,Kim2017b}. One can estimate this expectation value by preparing a new state $|\psi_O\rangle$, which can be expressed as
\begin{equation}
    |\psi_O\rangle = U_{(n-1),n}\cdots U_{2,3}\,U_{1,2}|0\ldots 0\rangle, \label{eq:compressed_circuit}
\end{equation}
where $U_{i,i+1}$ represents a unitary that acts on two groups of qubits, labeled by $i$ and $i+1$~\cite{Borregaard2019}, and the support of the observable $O$ lies in the union of the  last two groups of qubits; see Fig.~\ref{fig:nrqc_ex} for an illustration.

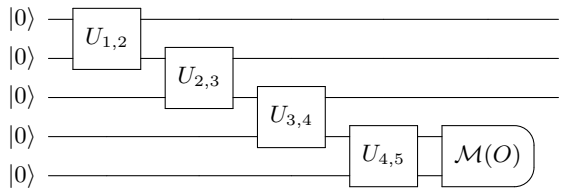
\begin{figure}[h]
\centering
$
    \Qcircuit @C=1em @R=.7em {
\lstick{\ket{0}}& \multigate{1}{U_{1,2}} & \qw & \qw & \qw & \qw &\qw \\
\lstick{\ket{0}}& \ghost{U_{1,2}}& \multigate{1}{U_{2,3}} &\qw &\qw & \qw &\qw\\
\lstick{\ket{0}}& \qw & \ghost{U_{2,3}} & \multigate{1}{U_{3,4}} &\qw &\qw &\qw \\
\lstick{\ket{0}}& \qw & \qw & \ghost{U_{3,4}} & \multigate{1}{U_{4,5}} &  \multimeasureD{1}{\mathcal{M}(O)} \\
\lstick{\ket{0}}& \qw & \qw & \qw & \ghost{U_{3,4}} &\ghost{\mathcal{M}(O)}\\
}
$\caption{A circuit diagram representing the state preparation of Eq.~\eqref{eq:compressed_circuit} for $n=5$ followed by the measurement of an observable $O$. 
}
    \label{fig:nrqc_ex}
\end{figure}

\begin{figure*}[t]
    \centering
    $
    \Qcircuit @C=1em @R=.7em {
\lstick{\ket{0}}& \multigate{1}{U_{1,2}} & \push{\ket{0}} & \qswap & \multigate{1}{U_{2,3}} &\push{\ket{0}} & \qswap &\qw &\cdots &\push{\ket{0}} & \qswap&\qw & \multigate{1}{U_{n-1, n}} & \qw & \multimeasureD{1}{\mathcal{M}(O)}\\
\lstick{\ket{0}}& \ghost{U_{1,2}} & \qw & \qswap \qwx & \ghost{U_{1,2}} &\qw & \qswap \qwx & \qw & \cdots &  &\qswap \qwx & \qw & \ghost{U_{n-1,n}} & \qw & \ghost{\mathcal{M}(O)}
}
$
    \caption{A space-efficient implementation of Eq.~\eqref{eq:compressed_circuit}. Compared to Fig.~\ref{fig:nrqc_ex}, the number of requisite qubits is reduced to a constant.}
    \label{fig:compressed_further}
\end{figure*}
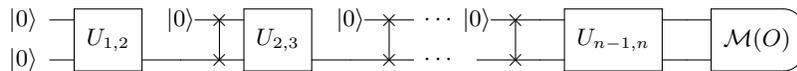

An important property of convolutional circuits is that the state preparation of Eq.~\eqref{eq:compressed_circuit} can be implemented using a constant number of qubits (independent of $n$), provided that one can reset the qubits at will; see Fig.~\ref{fig:compressed_further}.
However, as explained above, the experimental implementation of reset remains challenging, which motivates the search for alternatives. By applying the rewinding protocol to the preparation circuit truncated after the gate $U_{m-1,m}$, we obtain the circuit diagram in Fig.~\ref{fig:nrqc_recycled}. The rewound circuit prepares the same state on qubit $m$ as the original. In addition, though, the first few qubits are returned to the $|0\rangle$ state up to an error exponentially small in $m$: 
\begin{equation}
    \langle 0^k| \rho_{[1,k]} |0^k\rangle  \geq 1- c\exp(-\alpha(m-k)) \label{eq:error_decay}
\end{equation}
where $c$ and $\alpha$ are numerical constants. For a sufficiently large $m$, the error becomes negligible. Therefore, one can inject these freshly reset qubits into the remainder of the preparation circuit.

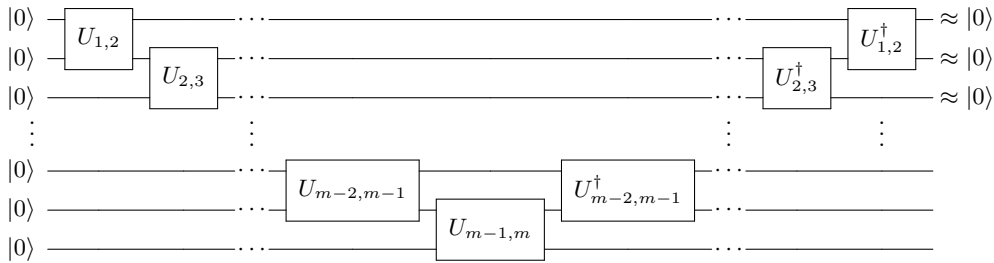
\begin{figure*}[t]
\centering
$
    \Qcircuit @C=.7em @R=.7em {
\lstick{\ket{0}}& \multigate{1}{U_{1,2}} & \qw & \qw & \cdots &  &\qw & \qw &\qw & \qw &\cdots &  &  \qw  & \multigate{1}{U_{1,2}^{\dagger}} &\qw &  &\approx \ket{0}\\
\lstick{\ket{0}}& \ghost{U_{1,2}}& \multigate{1}{U_{2,3}} &\qw &\cdots &  &\qw & \qw &\qw & \qw &\cdots & & \multigate{1}{U_{2,3}^{\dagger}} &\ghost{{U_{1,2}^{\dagger}}}&\qw &  &\approx \ket{0}\\
\lstick{\ket{0}}& \qw & \ghost{U_{2,3}} & \qw & \cdots & &\qw & \qw &\qw & \qw &\cdots &  & \ghost{{U_{2,3}^{\dagger}}} &\qw&\qw &  &\approx \ket{0}\\
\lstick{\vdots}&  &  &  & \vdots  & & & & & & \vdots & & & \vdots & & & \\
&  &  &  &  &  & & & & & & & & & & &\\
\lstick{\ket{0}}& \qw & \qw & \qw & \cdots &  & \multigate{1}{U_{m-2, m-1}} &\qw &\multigate{1}{U_{m-2,m-1}^{\dagger}} &\qw &\cdots & &\qw & \qw & \qw \\
\lstick{\ket{0}}& \qw & \qw & \qw & \cdots &  & \ghost{U_{m-2,m-1}} & \multigate{1}{U_{m-1,m}} &\ghost{U_{m-2,m-1}^{\dagger}} &\qw & \cdots & &\qw & \qw & \qw\\
\lstick{\ket{0}}& \qw & \qw & \qw & \cdots &  & \qw & \ghost{U_{m-1,m}} &\qw & \qw & \cdots & &\qw & \qw & \qw\\
}
$\caption{Applying the rewinding protocol (see Section~\ref{sec:protocol}) to Eq.~\eqref{sec:protocol}.}
    \label{fig:nrqc_recycled}
\end{figure*}

The rest of this section will be devoted to explaining why Eq.~\eqref{eq:error_decay} is true. A sketch of the proof will be presented in Section~\ref{sec:nrqc_proof}.
We also report on a numerical experiment that corroborates the proof in Section~\ref{sec:nrqc_exp}, focusing on extracting a precise estimate on the average number of iterations one needs to bound the fidelity to a level that is available in realistic devices, namely $0.01$. 

\subsection{Fidelity guarantee}
\label{sec:nrqc_proof}
In this section, we sketch the proof of Eq.~\eqref{eq:error_decay}, deferring the details to Appendix~\ref{appendix:transfer_matrix}. To explain the main idea, it will be convenient to replace the circuit diagram of Fig.~\ref{fig:nrqc_recycled} with a ``process'' diagram in which the unitary gate $U_{i,i+1}$ is replaced by a superoperator $\mathcal{U}_{i}(\cdot) = U_{i,i+1}(\cdot) U_{i,i+1}^{\dagger}$. For simplicity, suppose we trace out all but the first qubit; see Fig.~\ref{fig:nrqc_recycled_process}.

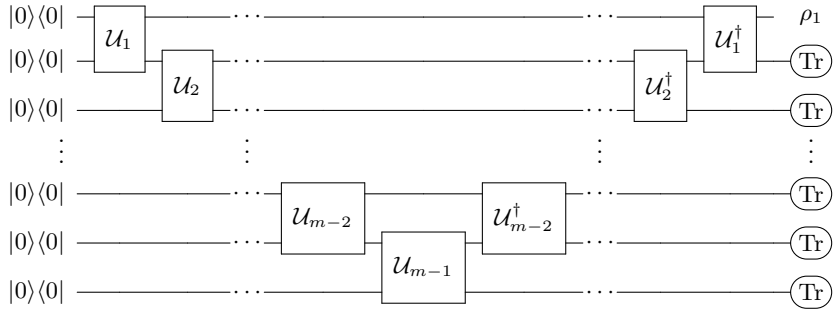
\begin{figure*}[t]
\centering
    $\Qcircuit @C=.7em @R=.7em {
\lstick{|0\rangle\langle0|}& \multigate{1}{\mathcal{U}_1} & \qw & \qw & \cdots &  &\qw & \qw &\qw & \qw &\cdots &  &  \qw  & \multigate{1}{\mathcal{U}_1^{\dagger}} &\qw & \rho_1 \\
\lstick{|0\rangle\langle0|}& \ghost{\mathcal{U}_1}& \multigate{1}{\mathcal{U}_2} &\qw &\cdots &  &\qw & \qw &\qw & \qw &\cdots & & \multigate{1}{\mathcal{U}_2^{\dagger}}  &\ghost{{\mathcal{U}_1^{\dagger}}}&\qw &\measure{\text{Tr}}\\
\lstick{|0\rangle\langle0|}& \qw & \ghost{\mathcal{U}_2} & \qw & \cdots & &\qw & \qw &\qw & \qw &\cdots &  & \ghost{{\mathcal{U}_2^{\dagger}}} &\qw&\qw &\measure{\text{Tr}}\\
\lstick{\vdots}&  &  &  & \vdots  & & & & & & \vdots & & & & & \vdots \\
&  &  &  &  &  & & & & & & & & & &\\
\lstick{|0\rangle\langle0|}& \qw & \qw & \qw & \cdots &  & \multigate{1}{\mathcal{U}_{m-2}} &\qw &\multigate{1}{\mathcal{U}_{m-2}^{\dagger}} &\qw &\cdots & &\qw & \qw & \qw &\measure{\text{Tr}}\\
\lstick{|0\rangle\langle0|}& \qw & \qw & \qw & \cdots &  & \ghost{\mathcal{U}_{m-2}} & \multigate{1}{\mathcal{U}_{m-1}} &\ghost{\mathcal{U}_{m-2}^{\dagger}} &\qw & \cdots & &\qw & \qw & \qw &\measure{\text{Tr}}\\
\lstick{|0\rangle\langle0|}& \qw & \qw & \qw & \cdots &  & \qw & \ghost{\mathcal{U}_{m-1}} &\qw & \qw & \cdots & &\qw & \qw & \qw &\measure{\text{Tr}}\\
}$
\caption{Circuit in Fig.~\ref{fig:nrqc_recycled} in the density matrix picture. Here $\mathcal{U}_i$ is a quantum channel that applies a unitary gate $U_{i,i+1}$. After tracing out all but the first qubit, we obtain a density matrix $\rho_1$.}
    \label{fig:nrqc_recycled_process}
\end{figure*}

Our argument will be directed at analyzing \emph{transfer operators} associated to the circuit, which will describe how portions of the rewinding circuit affect the rest. Consider a sequence of quantum channels $(\Phi_{i})$ for $i=1,\ldots, m$, whose action on a state $\rho$ is defined as follows.
\begin{equation}    
\begin{aligned}
\begin{array}{c}
\Qcircuit @C=.7em @R=.7em {
     \rho &  & \multigate{1}{\mathcal{U}_{i}} &\qw & \multigate{1}{\mathcal{U}_i^{\dagger}} &\qw &\Phi_i(\rho) \\  &
    \lstick{|0\rangle\langle 0|} & \ghost{\mathcal{U}_i} &  \gate{\Phi_{i+1}} &\ghost{\mathcal{U}_i^{\dagger}} &\qw & \measure{\text{Tr}}
    }
\end{array}
\end{aligned}
\end{equation}
Because $\Phi_{i+1}$ is a composition of quantum channels, so is $\Phi_i$. The transfer operator $\mathcal{T}_{i+1\to i}$ is the transformation taking the channel $\Phi_{i+1}$ to $\Phi_i$. 

The transfer operator has several useful properties. First, it maps the identity channel to the identity channel. In particular, viewing the transfer operator as a linear operator, the identity channel becomes an eigenvector with an eigenvalue $1$. Moreover, the transfer operator has an eigenvalues with modulus less or equal to $1$. Therefore, generically, by applying the transfer operator sufficiently many times, all the non-identity components decay away. What is left is a quantum channel that is exponentially close to the identity channel. We make this argument more precise in Appendix~\ref{appendix:transfer_matrix}, proving that
\begin{equation}
    \| \Phi_1 - \mathcal{I} \|_{\diamond}= \mathcal{O}(\lambda^{n}), \label{eq:main_result}
\end{equation}
where $\lambda$ is the modulus of the second largest eigenvalue of $\mathcal{T}_{i+1,i}$, maximized over all $i$, and $\|\cdots \|_{\diamond}$ is the diamond norm. In particular, we obtain
\begin{equation}
    \|\rho - \Phi_1(\rho)\|_1 = \mathcal{O}(\lambda^{n}).
\end{equation}

In practice, it could be difficult to estimate the precise value of $\lambda$. Without that knowledge, one cannot provide a quantitative bound on the fidelity of the rewinding process. However, there is a simple method that can overcome this problem. One can simply estimate the fidelity of the recycled state by measuring those qubits in the computational basis. Assuming that we expect the fidelity to be bounded from below by $1-\delta$, to estimate the fidelity up to an error of $\epsilon$, it suffices to take $\sim \delta(1-\delta)/\epsilon^2$ samples.

\subsection{Numerical experiment}
\label{sec:nrqc_exp}

While our proof firmly establishes that the rewinding protocol works for convolutional circuits, Eq.~\eqref{eq:main_result} is a conservative bound. To get a more precise quantitative estimate of the actual fidelity we can typically achieve, we have performed a numerical experiment using the Tensor Network library~\cite{tensornetwork}.

Consider Haar-random gates $U, V \in SU(4)$. We choose the gates in the circuit to be $U_{1,2} = U_{2,3} = \cdots = U_{n-2,n-1} = U$ and $U_{n-1,n}= V$. We applied the rewinding protocol to this circuit, making the length $n=150$. Fig.~\ref{fig:plot_average} presents the median (over $2000$ trials) of the overlap of the $i$-th qubit with $\ket{0}$ for different values of $i$. 

\begin{figure}[h]
\centering
\includegraphics[width=0.9\linewidth]{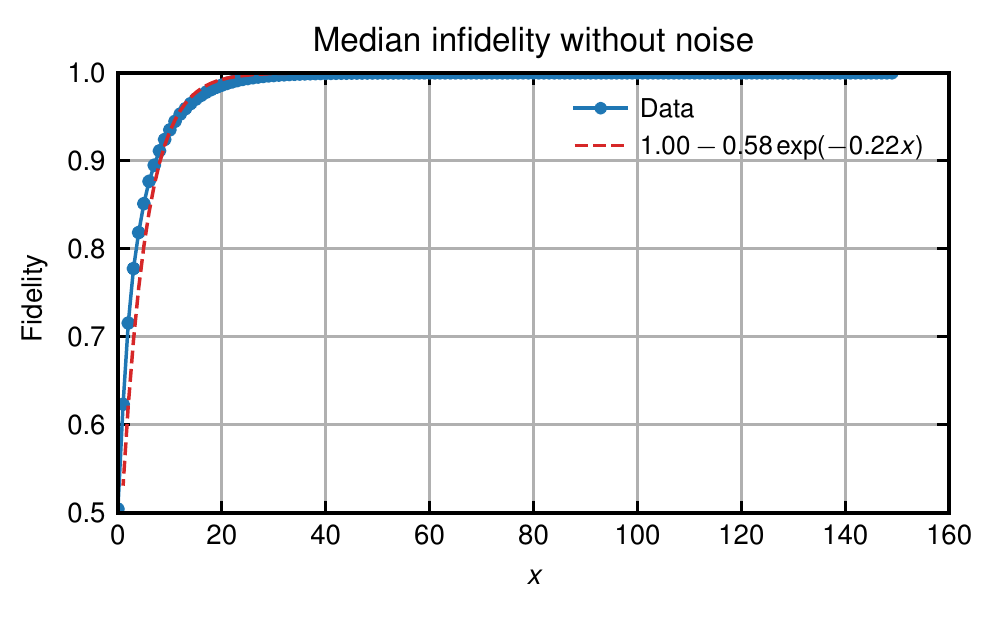}
\includegraphics[width=0.9\linewidth]{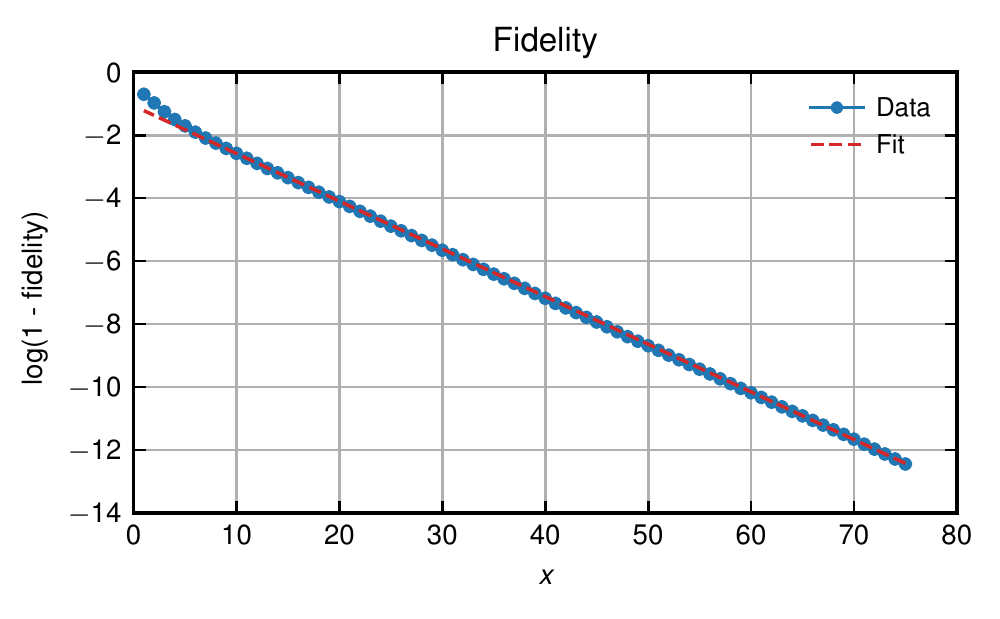}
    \caption{Median log fidelity plots. Here $x$ is $150- i$, where
    $i$ is the index of the qubit as defined above.}
    \label{fig:plot_average}
\end{figure}

As shown in Fig.~\ref{fig:plot_average}, our protocol yields a fidelity converging exponentially to $1$ with an average decay rate of 0.22. Typically, to reach a single-qubit preparation error of $0.01$ and $0.001$, we needed to apply the protocol of length 19 and 29 respectively.

\subsection{Robustness to noise}
\label{sec:noise}
Numerical evidence indicates that the rewinding protocol is also remarkably robust to noise. To model those imperfections, we redid the numerical experiment of Section~\ref{sec:nrqc_exp} but followed every gate with noise, specifically, each qubit being completely depolarized independently with probability $1\%$. The results are summarized in Fig.~\ref{fig:noise}. The fidelity of the qubits being reset converges to a quantity that is close to $1$ even in the presence of noise. In fact, the fidelity converges to a stationary value close to $0.99$, which suggests that the stationary value is $1- \mathcal{O}(p)$, where $p$ is the error probability per gate. Crucially, there is no indication that noise accumulates with increasing circuit depth. Moreover, up to the precision in our numerical experiment, the decay rate is $0.22$, which is the same value obtained \emph{in the absence of noise.} In Appendix~\ref{appendix:noise_resilience}, we have formulated a condition under which this robustness can be explained rigorously.

\begin{figure}[h]
\centering
\includegraphics[width=0.9\linewidth]{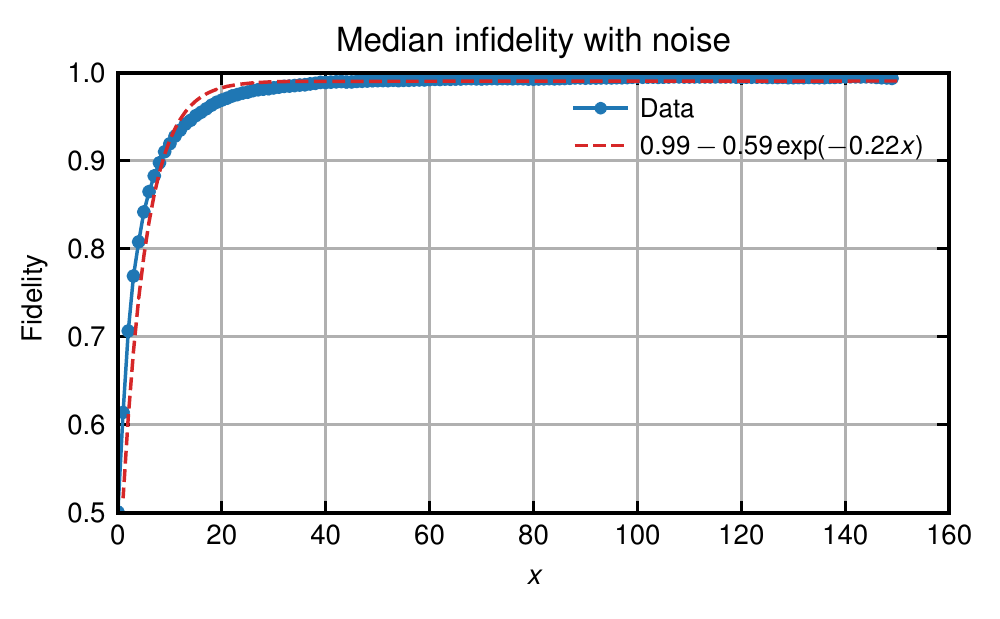}
\includegraphics[width=0.9\linewidth]{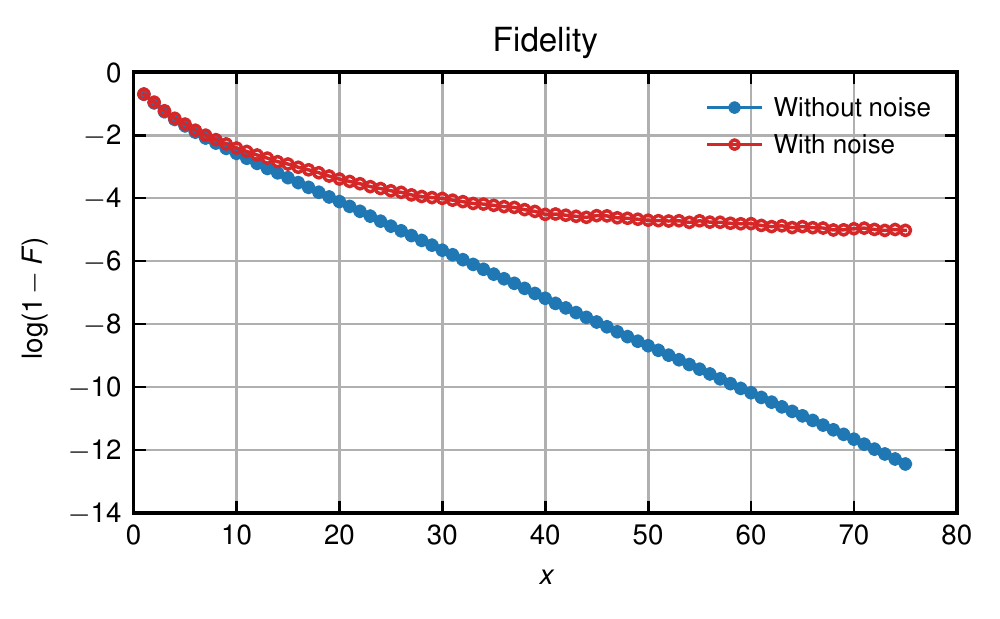}
    \caption{Median fidelity in the presence of noise, with the convention $x=150-i$, where
    $i$ is the index of the qubit as defined above.}
    \label{fig:noise}
\end{figure}

\section{Boosting the number of available qubits}
\label{sec:boost}
By repeatedly using the rewinding protocol, we can obtain an \emph{asymptotic} reduction in the number of qubits needed. Supppose, for concreteness, that we are equipped with an $n_q$ qubit quantum computer but somehow the circuit of interest consists of $n> n_q$ qubits. What should we do? Often, one is interested in estimating the expectation value of the observable up to some error, say $\epsilon$. The main question we address in this section is the size of the circuit one can simulate using $n$ qubits while ensuring that the error is below $\epsilon.$

The idea is to use the rewinding protocol recursively. Suppose, for instance, that we run the circuit only up to the first $n_1=n_q$ qubits and then apply the rewinding protocol, resetting the first $k_1$ qubits. We can feed those $k_1$ qubits in as new ancillae to be used for the remaining part of the circuit, incurring an error bounded by $\epsilon_1 = ce^{-\alpha(n_1-k_1)}$ in doing so. We are once again constrained by space so we apply a circuit over $n_2 = k_1 + 1$ qubits, namely the output of the first portion in addition to the new ancillae, and then apply the rewinding protocol. If we keep the first $k_2$ qubits, we incur an error on these qubits which can be bounded by
\begin{equation}
    \epsilon_2 \leq \epsilon_1 + ce^{-\alpha(n_2 - k_2)}.
\end{equation}
By repeating this $t$ times, the total error is bounded by:
\begin{equation}
    \epsilon_{\text{total}}  \leq \sum_{i=1}^{t}c (t-i+1)e^{-\alpha(k_{i-1} - k_i)},
    \label{eq:main_error_bound}
\end{equation}
where $k_0 = n_q$. The number of qubits effectively available for the original convolutional circuit is therefore $\sum_{i=0}^t k_i$. Since the size of a convolutional quantum circuit is proportional to the number of available qubits, the size is itself $\Theta(\sum_{i=0}^t k_i)$.

Let us choose $k_{i-1}  - k_i$ to be a constant, denoted $\Delta k$. To leading order, the $\Delta k$ that ensures $\epsilon_{\text{total}}$ to be below $\epsilon$ is
\begin{equation}
    \Delta k \approx \frac{\ln \left(t^2/\epsilon \right)  }{\alpha}.
\end{equation}
We can then choose $t\sim \frac{\alpha n_q}{\ln(1/\epsilon) \ln (n_q)}$ to get, to leading order, 
\begin{equation}
    n_{\text{circuit}} \approx \frac{\alpha n_q^2}{2\ln (1/\epsilon) \ln n_q}.
\end{equation}
Therefore, we get a nearly quadratic boost in the size of the circuit we can implement. For $\epsilon=0.01$ and the numerically observed average value of $\alpha=0.22$, we obtain $n_{\text{circuit}} \approx 0.024 n_q^2/ \ln n_q$. The rewinding protocol always increases the number of available qubits but this bound only detects the advantage at $n_q \sim 90$ qubits. A more optimized choice of $k_{i-1} - k_i$ could potentially yield a better bound. 

The preceding analysis allowed for the possibility that the induced error on the recycled qubits could be highly correlated. In practice, the error model might be close to being independent. While we do not have a general proof for such a statement, we have performed a preliminary numerical experiment that supports the hypothesis. If that is really the case, an \emph{extensive} error occurring on every qubit, say of fixed strength $\epsilon$, generically only gives rise to an error of order $\epsilon$ to the final answer~\cite{Kim2017,Kim2017a,Kim2017b}. Understanding to what degree this is true and what the ultimate limit of our protocol is an open problem left for future work.

As a practical matter, it is important to note that the procedure described above is highly parallelizable. There is no reason to wait until a given iteration of the rewinding protocol is complete before recycling qubits. As soon as a sufficiently high-fidelity qubit is prepared, one can use this qubit while running the remainder of the recycling protocol concurrently. The procedure discards the first $\Delta k$ qubits processed by the rewinding protocol, which will impose a corresponding delay at each level of recursion, but the delay is independent of the number of qubits $n_q$, depending only on the quality of the desired recycled qubits. In fact, the delay can even be eliminated by maintaining a buffer of $\Delta k$ fresh qubits. Therefore, the depth of the convolutional circuit implemented using the rewinding protocol to achieve a nearly quadratic savings in spatial cost can be made to be the same as the circuit executed with a physical supply of fresh qubits.

\section{Discussion}
\label{sec:discussion}
In this paper, we proposed a simple protocol for unitarily implementing a reset operation on a qubit. This protocol is particularly useful for reducing the spatial cost in implementing convolutional quantum circuits, which frequently appear in a large family of promising variational quantum circuits that can efficiently prepare physical ground states of interest~\cite{Kim2017,Kim2017a,Kim2017b,Borregaard2019,Barratt2020,Olund2020,Foss-Feig2020}. While implementation of such circuits using a physical reset operation has been carried out already~\cite{Foss-Feig2020}, the reported experimental system consists of a small number of qubits. By using our method, it should be possible to run these circuits on the existing larger quantum computers which are not equipped with physical reset operations.

Recently, Borregaard et al.~\cite{Borregaard2019} pointed out that convolutional circuits can be truncated to shorter circuits, leading to a reduction in spatial (and temporal) cost. Interestingly, the key parameter in their work that determines the savings is different from ours. In their case, the modulus of the second largest eigenvalue of the transfer matrix acting on the space of density matrices determines the optimal depth of the circuit beyond which it does not make sense to run the circuit any further. In our case, the savings are determined by the second largest eigenvalue of the transfer matrix acting on the \emph{space of completely positive trace-preserving maps}; see Appendix~\ref{appendix:transfer_matrix}. Understanding which of these two methods is better in practice is an open problem left for future work. 

There are several directions to improve our method. First, one may be able to devise a method to efficiently certify that the newly prepared ancilla states have large overlaps with the $|0\rangle$ state. One straightforward way to do this is to measure the state of the qubit directly, thus  gathering the statistics. However, this process can predict a lower fidelity because of the measurement error. A method to remove such measurement error in the spirit of randomized benchmarking~\cite{Emerson2005,RB_Knill} would be desirable. Second, one may be able to concatenate our method with the well-known algorithmic cooling method to further improve the purity of the ancilla qubits~\cite{Fernandez2004,Schulman2005}.

Our work focused on convolutional quantum circuits for which the local Hilbert space dimension is bounded. This is an important assumption because the bounds we derived in this paper were in terms of $2$-norm in some vector space. However, in order to convert a bound on this norm to a statement about the fidelity, one generally introduces a dimension-dependent factor. Therefore, if the local Hilbert space dimension is large, this bound may not be so useful. A similar problem was tackled in Ref.~\cite{Kim2017a}, by introducing a more refined notion of mixing called \emph{local mixing}. It will be interesting to understand if there is an analogous physically plausible local mixing condition in our context.

Lastly, the noise-resilience of our protocol warrants a further study. Specifically, if every gate in the circuit is replaced by an ideal gate followed by a local noise channel, would the fidelity of the prepared qubits be still $1-\mathcal{O}(\epsilon)$ where $\epsilon$ is the noise strength? Alternatively, would the fidelity scale as $1- \mathcal{O}(n\epsilon )$, where $n$ is the circuit depth? Our numerical experiment suggests the former, and our analysis in Appendix~\ref{appendix:noise_resilience} definitively shows that there is a sufficient condition under which one can rigorously prove the noise-resilience of our protocol. Whether this condition holds generically for convolutional quantum circuits is an important open problem that is left for future work.

\section*{Acknowledgments}
We thank Dave Bacon and Stefan Leichenauer for helpful conversations.
GA acknowledges the support of the Stanford Institute for Theoretical Physics and her coauthors. IK was supported by Google, the Simons Foundation It from Qubit Collaboration and by the Australian Research Council via the Centre of Excellence in Engineered Quantum Systems (EQUS) project number CE170100009. PH acknowledges the support of Gooogle, as well as AFOSR (award FA9550-19-1-0369), CIFAR, and the Simons Foundation.

\bibliography{bib}

\appendix

\section{Transfer operator}
\label{appendix:transfer_matrix}

In this appendix, we will prove properties of the transfer operator introduced in Section~\ref{sec:nrqc_proof}. Let us first recall the definition of the transfer operator.  Without loss of generality, let $U: \cH_A\otimes \cH_B \to \cH_A\otimes \cH_B$ be a unitary on two quantum systems, $A$ and $B$. Let $\cU(\rho) = U\rho U^\dagger$ be the corresponding channel. The transfer operator is defined as follows:
\begin{align}
\cT_U[\Phi] = \Tr_{B}\circ \cU^\dagger \circ (\cI_A \otimes \Phi)\circ \cU \circ \cE
\label{eq:transfer_matrix}
\end{align}
where $\cE(\rho) = \rho \otimes \ket{0}_B\bra{0}_B$.

Clearly, $\mathcal{T}_U$ maps quantum channels to quantum channels. It will be helpful to be precise about the domain and range $\mathcal{T}_U$. In order to be able to introduce eigenvectors, it will be convenient to view $\mathcal{T}_U$ as an operator acting on the complex vector space $\Seq{\mathcal{CP}_{\cH}^{TP}}$ spanned by the CPTP maps $\mathcal{CP}_{\cH}^{TP}$. There is a natural inner product $\langle \cdot,\cdot \rangle$ on $\langle\mathcal{CP}_{\cH}^{TP}\rangle$ induced by the Hilbert space inner product on $\mathcal{H}$.

In order to study properties of $\mathcal{T}_U$, it will be convenient to define two norms: the ``Hilbert-Schmidt'' norm over $\Seq{\mathcal{CP}_{\cH}^{TP}}$, which we shall denote as $\|\cdot \|_2$, and an operator norm for the operators that act on $\Seq{\mathcal{CP}_{\cH}^{TP}}$, which we shall denote as $\|\cdot \|$. Formally, these norms are defined as follows.
\begin{equation}
    \|\Phi \|_2 := \sqrt{\langle \Phi, \Phi \rangle}
\end{equation}
and
\begin{equation}
    \|\mathcal{T}\| = \sup_{\Phi \in \Seq{\mathcal{CP}_{\cH}^{TP}}} \frac{\|\mathcal{T}[\Phi] \|_2}{\|\Phi \|_2}.
\end{equation}

As a first application, we first show that $\cT_U$ has eigenvalues with modulus less or equal to $1$.
\begin{theorem}
\label{thm:eigenvalue_bound}
Suppose $U$ is a unitary and $\cT_U$ is defined as in Eq.~\eqref{eq:transfer_matrix}. Let $\Phi
\in \Seq{\mathcal{CP}_{\cH}^{TP}}$ be a right eigenvector of $\cT_U$ with eigenvalue $\lambda$. Then $\Abs{\lambda} \leq 1$.
\end{theorem}

\begin{proof}
Without loss of generality, let $\Phi = 
\sum_{i\in I} \alpha_i \Psi_i$, for some $\alpha_i \in \C$ and $\Psi_i \in \mathcal{CP}_{\cH}^{TP}$ for some finite index set $I$. The index can be made to be finite because the vector space $\Seq{\mathcal{CP}_{\cH}^{TP}}$ is finite-dimensional. Thus $\cT_U^k[\Phi] = \sum_{i\in I} \alpha_i \cT_U^k[\Psi_i]$ by linearity, where  $\cT_U^k$ is the $k$-th repeated application of $\cT_U$. 

After applying $\cT_U$ $k$ times, we obtain:
\begin{equation}
    \|\cT_U^k[\Phi] \|_2 = |\lambda|^k \|\Phi \|_2.
\end{equation} 
The key point is that this quantity becomes unbounded if $|\lambda|>1$ in the $k\to \infty$ limit. This is a contradiction because, using triangle inequality, we get
\begin{equation}
\begin{aligned}
    \|\cT_U^k[\Phi] \|_{2} &\leq \sum_{i\in I} \|\alpha_i \cT_U^k[\Psi_i] \|_{2} \\
    &\leq \sum_{i\in I} |\alpha_i| \|\widetilde{\Psi}_{i} \|_2,
\end{aligned}
\label{eq:proof_temp}
\end{equation}
where $\widetilde{\Psi}_{i} \in \mathcal{CP}_{\cH}^{TP}$ because $\mathcal{T}_U^k$ maps  quantum channels to quantum channels. Importantly, the last line of Eq.~\eqref{eq:proof_temp} is bounded for all $k$ because the set of quantum channels itself is bounded. If $|\lambda|> 1$, there exists a finite $k\in \mathbb{N}$ such that
\begin{equation}
    |\lambda|^k \|\Phi\|_2 > \sum_{i\in I} |\alpha_i| \|\widetilde{\Psi}_{i} \|_2,
\end{equation}
which is a contradiction. Therefore, we conclude that $|\lambda| \leq 1$, completing the proof.
\end{proof}

From Theorem~\ref{thm:eigenvalue_bound}, we found that the eigenvalues of the right eigenvectors of $\cT_U$ must lie on $|\lambda| \leq 1$. Moreover, one can verify straightforwardly that the identity channel, denoted as $\mathcal{I}$, satisfies the following relation:
\begin{equation}
    \cT_U[\mathcal{I}] = \mathcal{I}.
\end{equation}
This means that $\mathcal{I}$ is a right eigenvector of $\cT_U$ with an eigenvalue $1$. Provided that this is the only right eigenvector whose eigenvalue has a modulus of $1$, the norm of the other right eigenvectors must be strictly less than $1$. In particular, after a repeated application of $\cT_U$, the norm in fact decays exponentially.

\begin{corollary}
Let $\mathcal{T}$ be a transfer operator. Suppose there is only one right eigenvector of $\cT$ that has an eigenvalue with modulus $1$. Let $\Phi$ be a CPTP map and define a sequence of CPTP maps as  $\Phi_0 := \Phi$ and $\Phi_i := \cT[\Phi_{i-1}]$. Then
\begin{equation}
    \|\Phi_k - \mathcal{I} \|_2 \leq \mathcal{O}(\Delta^k)\label{eq:corollary_exponential},
\end{equation}
where $0 < \Delta < 1$. Specifically, $\Delta$ is the modulus of the largest non-unit eigenvalue of $\mathcal{T}$.
\end{corollary}

\begin{proof}
Recall that $\cT[\cI] = \cI$. Therefore, if $\cT$ has only one right eigenvector, it must be $\cI$. Without loss of generality, 
\begin{equation}
    \cT = P + D,
\end{equation}
where $P$ is a projector onto the one-dimensional subspace spanned by $\cI$ and $D = \cT - P$. Let us first prove that $P[\Phi] = \mathcal{I}$. Since $\mathcal{T}_i$ maps a CPTP map to a CPTP map, 
\begin{equation}
    \lim_{n\to \infty} \mathcal{T}^n[\Phi] = P[\Phi]
\end{equation}
is also a CPTP map. However, the only multiple of the identity that is a CPTP map is the identity map itself. Therefore, $P[\Phi]= \mathcal{I}$. 

Without loss of generality, consider the decomposition of $\Phi$ into the linear combination of right eigenvectors of $\mathcal{T}$:
\begin{equation}
    \Phi = \sum_{i\in I}\alpha_i \Psi_i,
\end{equation}
where $\mathcal{T}[\Psi_i] = \lambda_i$. Since $P[\Phi] = \mathcal{I}$, 
\begin{equation}
    \Phi - \mathcal{I} = \sum_{\substack{i\in I, \\ |\lambda_i| < 1 }} \alpha_i \Psi_i.
\end{equation}

Using the fact that $\mathcal{T}[\mathcal{I}] = \mathcal{I}$,
\begin{equation}
\begin{aligned}
    \|\Phi_k -\mathcal{I} \|_2 &= \|\mathcal{T}^k[\Phi - I] \|_2 \\
    &= \|\sum_{\substack{i\in I,\\ |\lambda_i| < 1}} \alpha_i \Psi_i  \lambda_i^k \|_2 \\
    &\leq \Delta^k \sum_{\substack{i\in I, \\ |\lambda_i|<1}} |\alpha_i| \| \Psi_i \|_2
\end{aligned}
\end{equation}
Since we are working in a finite-dimensional vector space, each of the terms in the last line is bounded.
\end{proof}

Let us remark that the distance $\|\Phi_k -\mathcal{I} \|_2$ upper bounds a more operationally meaningful distance of $\| \Phi_k -\mathcal{I} \|_{\diamond}$ up to a dimension-dependent factor. Because the underlying Hilbert space is finite, the fact that this distance decays exponentially in $k$ remains the same. This completes the main claim in Section~\ref{sec:nrqc_proof}.

\section{Noise resilience}
\label{appendix:noise_resilience}
In this appendix, we formulate a sufficient condition under which the rewinding protocol becomes resilient to noise. The central concept is the notion of \emph{contractivity} of the transfer matrix. 
\begin{definition}
Let $\mathcal{T}$ be a linear map which maps a channel to another channel. $\mathcal{T}$ is contractive if 
\begin{equation}
    \|\mathcal{T}[\Phi_1 - \Phi_2] \|_2 \leq \gamma(\mathcal{T}) \| \Phi_1 - \Phi_2\|_2
\end{equation}
for some real number $\gamma(\mathcal{T}) < 1$. If such a $\gamma(\mathcal{T})$ exists, the smallest such number will be defined as the contraction coefficient of $\mathcal{T}$.
\label{definition:contraction}
\end{definition}

Below, we will show that if the transfer matrix in Eq.~\eqref{eq:transfer_matrix} is contractive, the rewinding protocol is resilient to noise. Specifically,  without loss of generality, consider the noisy version of the transfer matrix (cf. Eq.~\eqref{eq:transfer_matrix}):
\begin{equation}
    \mathcal{T}_{U, \delta}[\Phi] = \text{Tr}_B \circ \mathcal{U}^{\dagger}_{\delta} \circ (\mathcal{I}_A \otimes \Phi) \circ \mathcal{U}_{\delta} \circ \mathcal{E},
    \label{eq:transfer_matrix_noisy}
\end{equation}
where $\mathcal{U}$ and $\mathcal{U}^{\dagger}_{\delta}$ are noisy versions of $\mathcal{U}$ and $\mathcal{U}^{\dagger}$:
\begin{equation}
\begin{aligned}
    \|\mathcal{U}_{\delta} - \mathcal{U} \| &\leq \delta, \\
    \|\mathcal{U}_{\delta}^{\dagger} - \mathcal{U}^{\dagger} \| &\leq \delta.
\end{aligned}
\end{equation}
Our goal is to show that the fixed point of $\mathcal{T}_{U,\delta}$ is close to that of $\mathcal{T}_U$ up to a distance $\mathcal{O}(\delta/(1-\gamma))$.

For small $\delta$, $\mathcal{T}_{U,\delta}$ and $\mathcal{T}_U$ are close to each other, as we prove below.
\begin{lemma}
\label{lemma:bound_transfer_operator_noise}
\begin{equation}
    \|\mathcal{T}_{U,\delta} - \mathcal{T}_U\| \leq C\delta
\end{equation}
for some constant $C$.
\end{lemma}
\begin{proof}
Note the following chain of inequalities.
\begin{widetext}
\begin{align}
\begin{aligned}
\| \mathcal{T}_{U,\delta} - \mathcal{T}_U\| &=  \sup_{\Phi \in \Seq{\mathcal{CP}_{\cH}^{TP}}} \frac{\|(\mathcal{T}_{U,\delta} - \mathcal{T}_U)[\Phi] \|_2}{\|\Phi \|_2} \\
& =   \sup_{\Phi \in \Seq{\mathcal{CP}_{\cH}^{TP}}} \frac{\|\text{Tr}_B\circ (\mathcal{U}_{\delta}^{\dagger} - \mathcal{U}^{\dagger}) \circ (\mathcal{I}_A \otimes \Phi) \circ \mathcal{U}_{\delta} \circ \mathcal{E} + \text{Tr}_B\circ \mathcal{U}^{\dagger} \circ (\mathcal{I}_A \otimes \Phi) \circ (\mathcal{U}_{\delta} - \mathcal{U}) \circ \mathcal{E}\|_2}{\| \Phi\|_2}
\\
&\leq  \sup_{\Phi \in \Seq{\mathcal{CP}_{\cH}^{TP}}} \frac{\|\text{Tr}_B\circ (\mathcal{U}_{\delta}^{\dagger} - \mathcal{U}^{\dagger}) \circ (\mathcal{I}_A \otimes \Phi) \circ \mathcal{U}_{\delta} \circ \mathcal{E} \|_2}{\| \Phi\|_2}
+ \sup_{\Phi \in \Seq{\mathcal{CP}_{\cH}^{TP}}} \frac{\|\text{Tr}_B\circ \mathcal{U}^{\dagger} \circ (\mathcal{I}_A \otimes \Phi) \circ (\mathcal{U}_{\delta} - \mathcal{U}) \circ \mathcal{E} \|_2}{\| \Phi\|_2} \\
&\leq \|\text{Tr}_B \|\mathcal{I}_A \| \|\mathcal{E} \| \left(\|\mathcal{U}_{\delta}^{\dagger} - \mathcal{U}^{\dagger}\| \|\mathcal{U}_{\delta}\| + \|\mathcal{U}_{\delta} - \mathcal{U} \| \|\mathcal{U}^{\dagger} \|  \right) \\
&\leq \delta \|\text{Tr}_B \|\mathcal{I}_A \| \|\mathcal{E} \| \left( \|\mathcal{U}_{\delta}\| + \|\mathcal{U}^{\dagger} \|  \right).
\end{aligned}
\end{align}
\end{widetext}
Since the operators placed in the norm in the last line are bounded operators acting on a finite-dimensional vector space, their norms are bounded.
\end{proof}

\begin{prop}
Consider a sequence $\mathcal{T}_1, \mathcal{T}_2,\ldots, \mathcal{T}_k$ of transfer operators, each potentially defined by a different unitary transformation. Suppose $\mathcal{T}_i$ has a contraction coefficient of $\gamma_i$ and let $\gamma = \max_i \gamma_i$.

For any sequence of operators $\widetilde{\mathcal{T}}_1, \widetilde{\mathcal{T}}_2, \ldots, \widetilde{\mathcal{T}}_k$ which map channels to channels such that $\|\mathcal{T}_i - \widetilde{\mathcal{T}}_{i} \|\leq \delta$ for all $i \in \{1,\ldots, k\}$,
\begin{equation}
    \|(\widetilde{\mathcal{T}}_k \circ \ldots \circ \widetilde{\mathcal{T}}_1 -  \mathcal{T}_k \circ \ldots \circ \mathcal{T}_1)[\Phi ] \|_2 \leq \frac{C'\delta}{1-\gamma}
\end{equation}
for some constant $C'>0$.
\end{prop}
\begin{proof}
Consider the telescopic decomposition of $\widetilde{\mathcal{T}}_k \circ \ldots \circ \widetilde{\mathcal{T}}_1 -  \mathcal{T}_k \circ \ldots \circ \mathcal{T}_1$:
\begin{equation}
\begin{aligned}
    &\widetilde{\mathcal{T}}_k \circ \ldots \circ \widetilde{\mathcal{T}}_1 -  \mathcal{T}_k \circ \ldots \circ \mathcal{T}_1 \\ &= \sum_{i=1}^{k-1} \mathcal{T}_{[i+1,k]}\circ( \widetilde{\mathcal{T}}_i - \mathcal{T}_i)\circ \widetilde{\mathcal{T}}_{[0, i-1]},
\end{aligned}
\end{equation}
where $\mathcal{T}_{[i,j]}$ and $\widetilde{\mathcal{T}}_{[i,j]}$ is a sequential application of $\mathcal{T}_{k}$ for $k$ from $i$ to $j$. Here we assume that $j\geq i$ and also that $\mathcal{T}_0$ is an identity operation, returning every channel to itself.

Therefore, 
\begin{equation}
\begin{aligned}
&\|    \sum_{i=1}^{k-1} \mathcal{T}_{[i+1,k]}\circ(\widetilde{\mathcal{T}}_i - \mathcal{T}_i)\circ \widetilde{\mathcal{T}}_{[0, i-1]}[\Phi] \|_2\\
&\leq 
\sum_{i=1}^{k-1} \mathcal{T}_{[i+1,k]}\circ(\widetilde{\mathcal{T}}_i - \mathcal{T}_i)[\Phi_i] \|_2
\end{aligned}
\end{equation}
for some channel $\Phi_i$. Note that
\begin{equation}
    \|\widetilde{\mathcal{T}}_i[\Phi_i] - \mathcal{T}_i[\Phi_i]\| \leq C'\delta
\end{equation}
for some constant $C'$ and that both $\widetilde{\mathcal{T}}_i[\Phi_i]$ and $\mathcal{T}_i[\Phi_i]$ are channels. Using the definition of the contraction coefficient, (see Definition~\ref{definition:contraction}) we obtain the following bound:
\begin{equation}
    \begin{aligned}
    &\|    \sum_{i=1}^{k-1} \mathcal{T}_{[i+1,k]}\circ(\widetilde{\mathcal{T}}_i - \mathcal{T}_i)\circ \widetilde{\mathcal{T}}_{[0, i-1]}[\Phi] \|_2\\
    &\leq C'\delta \sum_{i=1}^{k-1} \gamma^{k-i-1} \\
    &\leq \frac{C'\delta}{1-\gamma}
    \end{aligned}
\end{equation}
\end{proof}

\section{Examples of convolutional circuits}
\label{sec:convolutional_circuits_examples}
Convolutional circuits appear in many situations, including the encoding circuits for convolutional quantum codes~\cite{poulin2009quantum}. We focus on three examples: matrix product states~\cite{Schon2005,Barratt2020}, holographic quantum circuits~\cite{Kim2017,Kim2017a}, and deep multi-scale entanglement renormalization ansatz (DMERA)~\cite{Vidal2007,Evenbly2009,Kim2017b}.

\subsection{Matrix product states}
\label{sec:mps}
Matrix product state (MPS) is a many-body quantum state which can be expressed in the following form:
\begin{equation}
    |\Psi\rangle = \sum_{\{s \}} \text{Tr}[A_1^{(s_1)} A_2^{(s_2)} \cdots A_n^{(s_n)}] |s_1 s_2\ldots s_n\rangle,\label{eq:mps}
\end{equation}
where $\{ A_i^{(s_i)}\}$ is a set of matrices of dimension $\chi$ and $s_i \in \{ 0,\ldots, d-1\}$. Here $\chi$ is the dimension of the \emph{virtual space} $\mathcal{H}_v$ of the MPS and $d$ is the local Hilbert space dimension of the constituent particles, which has a tensor product form $\otimes_i \mathcal{H}_i$ where $\dim (\mathcal{H}_i)= d$. Formally, one can view $A_i^{(s_i)}$ as a matrix $A_i$ whose domain and image is $\mathcal{H}_v$ and $\mathcal{H}_v\otimes \mathcal{H}_i$. Specifically,
\begin{equation}
    A_i |\alpha\rangle =\sum_{\beta, s_i} (A_i)_{\beta \alpha}^{(s_i)}|\beta\rangle |s_i\rangle,
\end{equation}
where $\alpha, \beta \in \{0, \ldots, \kappa-1 \}$ and $s_i \in \{0,\ldots, d-1\}$.

One can consider the following special form of MPS:
\begin{equation}
    |\Psi\rangle = \sum_{\{s \}} \langle \phi_F | A_1^{(s_1)} A_2^{(s_2)} \cdots A_n^{(s_n)} |\phi_I\rangle |s_1 s_2\ldots s_n\rangle,\label{eq:mps_special}
\end{equation}
where $|\phi_I\rangle$ and $|\phi_F\rangle$ are $\kappa$-dimensional vectors. While a class of states in the form of Eq.~\eqref{eq:mps_special} forms a subclass of Eq.~\eqref{eq:mps}, this subclass is still capable of describing many physical states of interest, such as the W-state, cluster state, and GHZ state~\cite{Schon2005}. 

While the matrices appearing in Eq.~\eqref{eq:mps_special} can be arbitrary, up to normalization, such states can be reexpressed as
\begin{equation}
    |\Psi\rangle \otimes |\phi_F\rangle = \sum_{\{s \}} V_1^{(s_1)} V_2^{(s_2)} \cdots V_n^{(s_n)} |\phi_I\rangle |s_1 s_2\ldots s_n\rangle,
\end{equation}
where $V_i^{(s_i)}$ is an isometry from the virtual Hilbert space to the tensor product of the virtual Hilbert space and the local Hilbert space. Specifically, the action of this isometry, acting on a basis state $|\alpha\rangle$ in the virtual Hilbert space and $|s_i\rangle$ in the local Hilbert space, acts as 
\begin{equation}
    V_i|\alpha\rangle|s_i\rangle  = \sum_{\beta, s_i} (V_i)_{\beta \alpha}^{(s_i)} |\beta\rangle |s_i\rangle.
\end{equation}

Since $V_i$ is an isometry from $\mathcal{H}_v$ to $\mathcal{H}_v \otimes \mathcal{H}_i$, one can rewrite it as
\begin{equation}
    V_i = U_i(|0\rangle \otimes I_v),
\end{equation}
for some unitary acting on $\mathcal{H}_v \otimes \mathcal{H}_i$, where $|0\rangle \in \mathcal{H}_i$ and $I_v$ is the identity operator acting on $\mathcal{H}_v$. Viewed this way, the state in Eq.~\eqref{eq:mps_special} can be created by first preparing the following state:
\begin{equation}
    |\phi_I\rangle \otimes \underbrace{|0\rangle \otimes \ldots \otimes |0\rangle}_{n \text{ copies}}
\end{equation}
and applying $U_n$ to $U_1$ sequentially. 

The resulting circuit, up to a relabeling, has a convolutional structure. To see why, it is convenient to view the virtual Hilbert space as a $\lceil \log_d \chi \rceil$-qudit system. The key idea is to convert the following tensor representation of $U_i$ 
\begin{equation}
    \Qcircuit @C=1em @R=.7em {
\mathcal{H}_i& & \multigate{1}{U_i} & \qw & \mathcal{H}_i  \\ \mathcal{H}_v& & \ghost{U_i}& \qw & \mathcal{H}_v
}\label{eq:mps_unitary}
\end{equation}
to the following form.
\begin{equation}
    \Qcircuit @C=1em @R=.7em {
\mathcal{H}_i& & \multigate{1}{\widetilde{U}_i} & \qw & \mathcal{H}_v  \\ \mathcal{H}_v& & \ghost{\widetilde{U}_i}& \qw & \mathcal{H}_i
}\label{eq:mps_unitary_flipped}
\end{equation}
Formally, $\widetilde{U}_i$ is $U_i$ followed by this map:
\begin{equation}
    |\alpha\rangle |s_i\rangle \to |s_i\rangle |\alpha\rangle
\end{equation}
for all $\alpha \in \{0,\ldots, \kappa-1\}$ and $s_i \in \{0,\ldots, d-1\}$. Since this is a permutation of qudits, it can be realized unitarily. Therefore, $\widetilde{U}_i$ is simply a unitary acting on $\lceil \log_d \chi \rceil +1$ qudits.

Using $\widetilde{U}_i$, the convolutional structure of MPS becomes evident. We provide the following example, which is a MPS over $3$ qudits.
\begin{equation}
    \Qcircuit @C=1em @R=.7em {
        |0\rangle &  & \qw & \qw & \multigate{1}{\widetilde{U}_1} &\qw & |\phi_F\rangle\\
    |0\rangle &  & \qw & \multigate{1}{\widetilde{U}_2} &\ghost{\widetilde{U}_1} &\qw & \\
|0\rangle& & \multigate{1}{\widetilde{U}_3} & \ghost{\widetilde{U}_2}&\qw &\qw & \\ 
|\phi_I\rangle&  & \ghost{\widetilde{U}_3}& \qw &\qw &\qw &} \label{eq:mps_example}
\end{equation}
Generalization to $n$-qudit MPS is straightforward.

\subsection{Holographic quantum circuits}
\label{sec:holographic}

Holographic quantum circuits~\cite{Kim2017,Kim2017a,Borregaard2019,Foss-Feig2020} generalize MPS to higher dimensions by replacing a qudit with a \emph{collection} of qudits and imposing a spatial structure within $U_i$. These circuits can be used as ansatz to simulate a $\mathcal{D}$-dimensional quantum many-body systems by applying a sequence of local gates on a set of qudits arranged on a $(\mathcal{D}-1)$-dimensional lattice. 

For concreteness, let us describe a $\mathcal{D}=2$-dimensional example. Without loss of generality, suppose the physical system of interest consists of qudits with local Hilbert space dimension $d$, arranged on a square lattice of size $\ell_x \times \ell_y$. Each qudit shall be specified by its coordinate $(x, y)$, where $x\in [\ell_x] = \{0,\ldots, \ell_x-1 \}$ and $y\in [\ell_y]=\{0,\ldots, \ell_y-1 \}$. To describe our circuit, we will also need to introduce a total of $\ell_y$ ``virtual'' qudits, each of which has local Hilbert space dimension of $\chi$. We shall specify them with an index $y\in [\ell_y]$.

Just like in the MPS example, a holographic quantum circuit consists of a sequence of unitaries $\{ U_i\}$ that act on the virtual qudits and a set of qudits at $x=i$. The state is prepared by first initializing all the qudits as follows
\begin{equation}
    \overbrace{\left(\bigotimes_{y=0}^{\ell_y-1}|0\rangle_y\right)}
    ^{\text{Virtual qudits}} \bigotimes \overbrace{\underbrace{\left( \bigotimes_{y=0}^{\ell_y-1}\right) \bigotimes \ldots \bigotimes  \left( \bigotimes_{y=0}^{\ell_y-1}\right)}_{\ell_x \text{copies}}}^{\text{Physical qudits}}
\end{equation}
and then applying $U_i$ from $i=\ell_x - 1$ to $0$.

Importantly, $U_i$ is a \emph{finite-depth local quantum circuit}, meaning that it is a $\mathcal{O}(1)$-depth quantum circuit consisting of gates that act on qudits whose $y$-coordinates are close to each other. An example of such $U_i$ for a depth of $3$ and $\ell_y=3$ is described below.
\begin{equation}
\Qcircuit @C=1em @R=.7em 
{
v:0 & &\qw &\multigate{1}{} &\qw &\multigate{1}{} & \qw \\
p:0 & &\qw &\ghost{} & \multigate{1}{} & \ghost{} &\qw \\
v:1 & &\qw &\multigate{1}{} & \ghost{} &  \multigate{1}{} & \qw\\
p:1 & &\qw &\ghost{} & \multigate{1}{} & \ghost{} & \qw\\
v:2 & &\qw &\multigate{1}{} & \ghost{} & \multigate{1}{} & \qw\\
p:2 & &\qw &\ghost{} &\qw &\ghost{} & \qw
}
\end{equation}
Here $v:y$ is the virtual qudit labeled by $y$ and $p:y$ is the physical qudit labeled by $(i, y)$.

To restore the convolutional structure, we can swap each virtual qudit with a physical qudit with the same $y$-coordinate, similar to how we did in Eqs.~\eqref{eq:mps_unitary} and~\eqref{eq:mps_unitary_flipped}.

\subsection{DMERA}
\label{sec:dmera}
DMERA is a multi-scale ansatz that can well-approximate many-body wavefunctions of scale-invariant  quantum many-body systems~\cite{Vidal2007,Kim2017b}. This ansatz can be prepared in a recursive manner. Without loss of generality, suppose we are interested in simulating a scale-invariant system in $\mathcal{D}$ spatial dimensions. The DMERA ansatz over $2^{n\mathcal{D}}$ qudits can be constructed by first preparing an ansatz over the qudits arranged on a lattice of size $\underbrace{2^{n-1} \times \ldots \times 2^{n-1}}_{\mathcal{D}}$, embedding the constituent qudits into a lattice of size $\underbrace{2^{n} \times \ldots \times 2^{n}}_{\mathcal{D}}$ with a lattice spacing of $(\underbrace{2,\ldots, 2}_{\mathcal{D}})$, filling the remaining lattice sites by qudits initialized in the $|0\rangle$ state, and then applying a finite-depth local quantum circuit. Without loss of generality, we will assume that the depths of these local circuits are bounded by $D$.

For the purpose of estimating the expectation value of a local observable $O$, one can remove some of the spurious gates that do not affect its expectation value. For instance, suppose the DMERA ansatz $|\Psi\rangle$ can be represented as $u|\Psi'\rangle$ for some gate $u$. If $u$ commutes with $O$, then the following identity holds:
\begin{equation}
    \langle \Psi| O |\Psi\rangle = \langle \Psi'| O |\Psi'\rangle.
\end{equation}
By applying this identity recursively, one can remove gates used in preparing $|\Psi\rangle$ until this procedure no longer yields a smaller circuit. The resulting minimal circuit is referred to as the \emph{past causal cone} of the observable $O$~\cite{Evenbly2009}.

Once the past causal cone is obtained, the expectation value of $O$ can be expressed as follows:
\begin{equation}
    \langle O\rangle = \text{Tr}[ \Phi_n\circ \cdots \circ \Phi_1(\rho) O ],\label{eq:transfer_operator_dmera}
\end{equation}
where $\Phi_i$ is a completely-positive trace preserving(CPTP) map called the \emph{transfer operator} acting on $\mathcal{O}(D^{\mathcal{D}})$ qudits, whose exact form can be deduced from the DMERA circuit; see Ref.~\cite{Evenbly2009,Kim2017b} for details. 

By the Stinespring dilation theorem~\cite{Stinespring}, every CPTP map can be viewed as an isometry followed by a partial trace. Moreover, this isometry can be directly deduced from the DMERA circuit, as explained in Ref.~\cite{Evenbly2009,Kim2017b}. Therefore, Eq.~\eqref{eq:transfer_operator_dmera} can be viewed as an expectation value of $O$ evaluated with respect to a state created by applying a sequence of isometries $\{ V_i\} $ which are the dilations of the CPTP maps $\{\Phi_i \}$. Then, by the analysis of Appendix~\ref{sec:mps}, the expectation value of $O$ can be viewed as an expectation value evaluated with respect to a state prepared by a convolutional circuit.

\end{document}